\documentclass[bj,preprint]{imsart}

\RequirePackage{amsthm,amsmath,amsfonts,amssymb}
\RequirePackage[numbers]{natbib}
\RequirePackage[colorlinks,citecolor=blue,urlcolor=blue]{hyperref}
\RequirePackage{graphicx}
\usepackage{enumerate} 
\usepackage{comment}
\usepackage{listings}
\usepackage{amsmath}
\usepackage{amsthm}
\usepackage{tikz}
\usepackage{caption}
\usepackage{array}
\usepackage{multirow}
\usepackage{mdwtab}
\usepackage{eqparbox}
\usepackage{multicol}
\usepackage{amsfonts}
\usepackage{multirow,bigstrut,threeparttable}
\usepackage{amsthm}
\usepackage{array}
\usepackage{bbm}
\usepackage{subfigure}
\usepackage{epstopdf}
\usepackage{mdwtab}

\usepackage{eqparbox}
\usepackage{latexsym}
\usepackage{bm}
\usepackage{amssymb}
\usepackage{graphicx}
\usepackage{mathrsfs}
\usepackage{epsfig}
\usepackage{psfrag}
\usepackage{algorithm}
\usepackage{algpseudocode}
\usepackage{stfloats}
\usepackage{algorithm}
\usepackage{mathrsfs}
\usepackage{graphicx}
\usepackage{tabularx,booktabs,siunitx,caption}
\newcolumntype{C}{>{\centering\arraybackslash}X}
\startlocaldefs
\numberwithin{equation}{section}
\theoremstyle{plain}
\newtheorem{theorem}{Theorem}
\newtheorem{example}{Example}
\newtheorem{lemma}{Lemma}
\newtheorem{remark}{Remark}
\newtheorem{corollary}{Corollary}
\newtheorem{proposition}{Proposition}
\newtheorem{definition}{Definition}

\def \bE {\mathbb{E}}
\def \bR {\mathbb{R}}

\def \var {\mathsf{Var}}

\usepackage{xspace}

\newcommand\itemaone{\item[($\mathcal A1$):]}
\newcommand\itematwo{\item[($\mathcal A2$):]}

\newcommand\itemtaone{\item[($\tilde{\mathcal A}1$):]}
\newcommand\itemtatwo{\item[($\tilde{\mathcal A}2$):]}

\newcommand{\floor}[1]{{\left\lfloor {#1} \right \rfloor}}

\newcommand{\Prob}{\mathbb{P}}

\newcommand{\TV}{{\sf TV}}
\newcommand{\KL}{{\sf KL}}

\newcommand{\Bern}{\mathrm{Bern}}
\newcommand{\Poi}{\mathsf{Poi}}
\newcommand{\Nor}{\mathsf{N}}

\newcommand{\Indc}{\mathbf{1}}

\definecolor{myblue}{rgb}{.8, .8, 1}
\definecolor{mathblue}{rgb}{0.2472, 0.24, 0.6} % mathematica's Color[1, 1--3]
\definecolor{mathred}{rgb}{0.6, 0.24, 0.442893}
\definecolor{mathyellow}{rgb}{0.6, 0.547014, 0.24}

\newcommand{\calE}{{\mathcal{E}}}

\newcommand{\calX}{{\mathcal{X}}}

\newcommand{\dpex}{p_{\mathrm{EX}}}
\newcommand{\dpmh}{p_{\mathrm{MH}}}

\newcommand{\pex}{P_{\mathrm{EX}}}
\newcommand{\pmh}{P_{\mathrm{MH}}}

\usepackage{cleveref}
\crefname{lemma}{Lemma}{Lemmas}
\Crefname{lemma}{Lemma}{Lemmas}
\crefname{thm}{Theorem}{Theorems}
\Crefname{thm}{Theorem}{Theorems}
\crefformat{equation}{(#2#1#3)}

\theoremstyle{remark}

\endlocaldefs

\begin{document}

\begin{frontmatter}
\title{On the Theoretical Properties of the Exchange Algorithm}
%\title{A sample article title with some additional note\thanksref{t1}}
\runtitle{On the Theoretical Properties of the Exchange Algorithm}
%\thankstext{T1}{A sample additional note to the title.}

\begin{aug}
%%%%%%%%%%%%%%%%%%%%%%%%%%%%%%%%%%%%%%%%%%%%%%
%%Only one address is permitted per author. %%
%%Only division, organization and e-mail is %%
%%included in the address.                  %%
%%Additional information can be included in %%
%%the Acknowledgments section if necessary. %%
%%%%%%%%%%%%%%%%%%%%%%%%%%%%%%%%%%%%%%%%%%%%%%
\author{\fnms{Guanyang } \snm{Wang}\ead[label=e1]{guanyang.wang@rutgers.edu}}
%%%%%%%%%%%%%%%%%%%%%%%%%%%%%%%%%%%%%%%%%%%%%%
%% Addresses                                %%
%%%%%%%%%%%%%%%%%%%%%%%%%%%%%%%%%%%%%%%%%%%%%%
\address{Department of Statistics and Biostatistics, Rutgers University,
\printead{e1}}
\end{aug}

\begin{abstract}
The exchange algorithm is one of the most popular extensions of the Metropolis--Hastings algorithm to sample from doubly-intractable distributions. However,  the theoretical exploration of the exchange algorithm is very limited. For example, natural questions like `Does exchange algorithm converge at a geometric rate?' or `Does the exchange algorithm admit a Central Limit Theorem?' have not been answered yet. In this paper, we study the theoretical properties of the exchange algorithm, in terms of asymptotic variance and convergence speed. We compare the exchange algorithm with the original Metropolis--Hastings algorithm and provide both necessary and sufficient conditions for the geometric ergodicity of the exchange algorithm. Moreover, we prove that our results can be applied to various practical applications such as location models, Gaussian models, Poisson models, and a large class of exponential families, which includes most of the practical applications of the exchange algorithm.  A central limit theorem for the exchange algorithm is also established.  Our results justify the theoretical usefulness of the exchange algorithm.
\end{abstract}



\end{frontmatter}
%%%%%%%%%%%%%%%%%%%%%%%%%%%%%%%%%%%%%%%%%%%%%%
%% Please use \tableofcontents for articles %%
%% with 50 pages and more                   %%
%%%%%%%%%%%%%%%%%%%%%%%%%%%%%%%%%%%%%%%%%%%%%%
%\tableofcontents

\section{Introduction}\label{sec: doubly intractable motivation}

Models with unknown normalizing constants arise frequently in many different areas. Examples include Ising models \cite{ising1925beitrag} in  statistical physics, autologistic models \cite{besag1972nearest} \cite{besag1974spatial} in spatial statistics, exponential random graph models \cite{robins2007introduction} in sociology, disease transmission models  \cite{potter2012estimating} in epidemiology, and so on. The corresponding statistical inference problem can  be formulated as follows. 

Suppose we were given data $x\in \mathcal X$ sampled from a family of probability densities (or probability mass functions) of the form:
\begin{equation}\label{eqn:doubly intractable}
p_\theta(x) = \frac{f_\theta(x)}{Z(\theta)}.
\end{equation}

We assume $f_\theta(x)$ can be easily evaluated but the normalizing function $Z(\theta) = \int_{\mathcal X}  f_\theta(x) dx$ is computationally intractable. Examples include:
\begin{example}[Ising Model]\label{eg:ising}
	Consider a graph $G = (V, E)$ with $n$ nodes, each vertex $i$ is assigned with a spin $\sigma_i$, which is either $1$ or $-1$. A spin configuration $\sigma \in \{-1, 1\}^n$ is an assignment of spins to all the graph vertices.  An Ising model  on G is defined by the following Boltzmann distributions over all possible configurations:
	\begin{equation}\label{eqn:ising}
	\Prob_\theta(\sigma) = \frac{e^{-\theta H(\sigma)}}{Z(\theta)},
	\end{equation}
	where $H(\sigma) = - \sum_{(i,j)\in E} J_{i,j}\sigma_i \sigma_j - M \sum_{i\in V} \sigma_i$ is  the Hamiltonian function, $J_{i,j}$ is  the interaction between spin $i$ and $j$, $M$ is  the magnetic moment, and $Z(\theta) = \sum_{\sigma} e^{-\theta H(\sigma)}$ is the partition function. As there are $2^n$ different possible spin configurations, the normalizing constant is usually computationally intractable for moderately large $n$.
\end{example}

\begin{example}[Exponential Random Graph Model]\label{eg:ergm}
	Exponential random graph models are a family of probability distributions on graphs. Let $\mathcal G_n$ be the set of all simple, undirect graphs without loops or multiple edges on $n$ vertices. Consider the following distribution on $\mathcal G_n$:
	\begin{equation}\label{eqn:ergm}
	\Prob_\theta(G = g) = \frac{e^{\theta s(g)}}{Z(\theta)},
	\end{equation}   
	where $s$ is a sufficient statistics defined on $\mathcal G_n$. This may be chosen as the degrees of the vertices, the number of edges, the number of triangles, or other sub-graph counts, $Z(\theta) = \sum_{g\in \mathcal G_n} e^{\theta s(g)}$ is the normalizing constant. As there are up to $2^{\binom{n}{2}}$ possible graphs, $Z(\theta)$ is also computationally intractable for moderately large $n$.
	
\end{example}

It is of natural interest to do inference on the parameter $\theta$. However, the classical route for statistical inference (maximum likelihood approach) can not be applied due to the intractability of $Z(\theta)$. Current frequentist solutions are mainly based on approximation methods such as  pseudo-likelihood approximation \cite{besag1974spatial},   MCMC-MLE \cite{geyer1991markov}, stochastic approximation \cite{younes1988estimation}.  Usually frequentist approaches are computationally efficient but do not have  theoretical guarantees. In fact, it is known that there are cases  these approximation methods perform poorly, see \cite{cucala2009bayesian} for discussion. 

In a Bayesian prospective, suppose  a prior $\pi(\theta)$ is adopted. The posterior can be formally calculated by $\pi(\theta| x) \propto \pi(\theta) p_\theta(x)$. Then a central part of Bayesian inference is to understand the posterior distribution. For example, if one is able to (asymptotically) draw samples from the posterior (usually by Markov-chain Monte Carlo  algorithms), then the distribution of any function $h(\theta)$ of interest can be  estimated by
$
\hat h= \frac{1}{N} \sum_{i= 1}^{N} h(\theta_i),
$
where $\theta_1, \cdots, \theta_N$ are samples drawn from $\pi(\theta|x)$.

However, the unknown normalizing function $Z(\theta)$ makes MCMC sampling pretty challenging. Consider a standard Metropolis--Hastings (MH) algorithm with proposal density $q$, in each iteration the acceptance probability is of the form:
\begin{equation}\label{eqn:MH_acceptance ratio}
\min\bigg(1, \frac{q(\theta',\theta)\pi(\theta'|x)}{q(\theta,\theta')\pi(\theta|x)}\bigg) =\min\bigg(1,\frac{q(\theta',\theta)\pi(\theta')f_{\theta'}(x)}{q(\theta,\theta')\pi(\theta)f_{\theta}(x)}\cdot \frac{Z(\theta)}{Z(\theta')}\bigg).
\end{equation}
This can not be directly computed as the ratio  $\frac{Z(\theta)}{Z(\theta')}$ is unknown. The posterior distribution $\pi(\theta|x)$ is often referred to as a \textbf{doubly-intractable} distribution as the Metropolis--Hastings algorithm is accurate only after infinity steps, and each iteration includes an infeasible calculation \cite{murray2012mcmc}.

One of the most popular methods to resolve this issue is the exchange algorithm \cite{murray2012mcmc} proposed by Murray et al. Roughly speaking, the exchange algorithm is a new MCMC algorithm which uses an auxiliary variable  at each step to estimate the unknown ratio $Z(\theta)/Z(\theta')$  (see Algorithm \ref{alg:exchange} for details). The algorithm is easy to implement and is asymptotically exact. 

The exchange algorithm is widely used in sampling from doubly-intractable distributions. However, there are very limited studies about its theoretical properties. One fundamental problem with the MCMC algorithm is its convergence rate. On the one hand, an a-priori bound on how long the chain should run to converge within any given accuracy would be helpful to guide practical uses. On the other hand, present theories show there are deep connections between the convergence rate and Markov-chain Central Limit Theorem. A chain with a sub-geometric convergence rate may fail to admit the Central Limit Theorem and  the estimator derived by Markov chain samples may even have infinite variance.

This motivates us to study the theoretical properties of the exchange algorithm.  Our main contributions include:
\begin{itemize}
	
	\item We prove several comparison-type results between the exchange algorithm and the original Metropolis--Hastings algorithm. Our results compare the exchange algorithm and the Metropolis--Hastings algorithm in terms of asymptotic variance and convergence rate. 
	\item We provide sufficient conditions to ensure the geometric ergodicity of the exchange algorithm. In particular, when the proposal distribution is symmetric, our assumptions to ensure the geometric ergodicity of the exchange chain are  weaker than the previous results, which allows us to study the convergence rate of the exchange algorithm on unbounded parameter space, and apply our results on many practical models.

	\item  We apply our theoretical results on a variety of practical examples such as location models, Ising models, exponential random graph models which include many of the practical applications of  exchange algorithms. Our results justify the theoretical usefulness of the exchange algorithm in practical situations. To our best knowledge, this is the first result to establish geometric ergodicity for the exchange algorithm on non-compact parameter space.
	\item We prove a Central Limit Theorem for the exchange algorithm given  it is geometrically ergodic. We also provide lower and upper bounds for the asymptotic variance of the exchange algorithm.
\end{itemize}

The remainder of the paper is organized as follows. In Section \ref{sec:preliminaries} we set up preliminary definitions, review current related results, and introduce the notations we used in this paper. Our main findings are stated and proved in Section \ref{sec:theoretical exchange algorithm}.  Section \ref{sec:conclusion} concludes this paper and provides further possible directions. 

We also briefly summarize our  theoretical results. In this paper we study the asymptotic variance and the convergence rate of the exchange algorithm, with an emphasis on its theoretical properties on practical models. Theorem \ref{thm:asymptotic variance} in Section \ref{subsec:asymptotic variance} shows the asymptotic variance of the original chain  is always no larger than the exchange algorithm. Our proof relies on the Peskun's ordering between the exchange algorithm and the original chain (Lemma \ref{lem:Peskun_order}). For convergence rate analysis, Theorem \ref{thm:lazy MH chain} shows variance bounding (or admitting a positive right spectral gap) of the original chain is a necessary condition for the geometric ergodicity of the exchange chain (this condition is not sufficient, see Example \ref{eg: exponential gamma} for a counterexample). Theorem \ref{thm:geometric ergodicity inherit, uniform case} and Theorem \ref{thm:Geometric ergodic inhirit} are the two main theorems concerning the inheritance of geometric ergodicity for the exchange algorithm. Theorem \ref{thm:geometric ergodicity inherit, uniform case} shows that if the original Metropolis--Hastings chain is geometrically ergodic, and the likelihood ratio function is uniformly non-negligible (see Definition \ref{def: non-negligible} for a precise definition), then the exchange algorithm is also geometrically ergodic. The condition in Theorem \ref{thm:geometric ergodicity inherit, uniform case} is usually satisfied when the parameter space has a compact closure, but often fails to  hold when the parameter space is unbounded. Theorem \ref{thm:Geometric ergodic inhirit} proves that the exchange algorithm can inherit the geometric ergodicity of the original algorithm under a much weaker condition when the proposal kernel is symmetric. Section \ref{subsec:applications of general state convergence} shows Theorem \ref{thm:Geometric ergodic inhirit} can be applied to many practical models, including location models, Poisson models, a large subset of exponential family models which contains ERGMs and Ising Models. 
The proof of Theorem \ref{thm:Geometric ergodic inhirit} relies on a `change of kernel' technique, which connects the geometric ergodicity results for Random-walk Metropolis--Hastings algorithms  \cite{mengersen1996rates}  \cite{roberts1996geometric} \cite{jarner2000geometric} with the Markov chain comparison results developed by Roberts and Rosenthal \cite{roberts2008variance}.
The `change of kernel' trick is new to the author's best knowledge, which may be of independent interest. 

For the experienced reader, here is a brief discussion between our results and the results developed by Andrieu
and Roberts \cite{andrieu2009pseudo}, and Andrieu and Vihola \cite{andrieu2015convergence} who studies the theoretical properties of the Pseudo-marginal MCMC algorithms, which is another popular approach to tackle the doubly-intractable distributions. The asymptotic variance results  in our paper (Lemma \ref{lem:Peskun_order}, Theorem \ref{thm:asymptotic variance}) are similar to Theorem 7 in \cite{andrieu2015convergence}. However, the exchange algorithm is dominated by the original Metropolis--Hastings algorithm in Peskun's ordering, but for Pseudo-marginal MCMC algorithms there is no such general ordering as its parameter space is defined as an enlarged product space. The conditions for convergence results proved in Theorem \ref{thm:geometric ergodicity inherit, uniform case} are similar to Theorem 8 in \cite{andrieu2009pseudo} and Section 3 in \cite{andrieu2015convergence}.  In the Pseudo-marginal MCMC papers, it is  required that the weight function is uniformly bounded.  In this paper, it is required that the likelihood ratio function is uniformly non-negligible.
When the proposal distribution is symmetric, our results and methods for studying the `geometric ergodicity inheritance' of the exchange algorithm are different from the previous results. For example, Theorem 38 of \cite{andrieu2015convergence} gives a condition on the `polynomial ergodicity' instead of `geometric ergodicity' of the Pseudo-marginal algorithms when the original chain is a Random-walk Metropolis Hastings chain. Moreover, we have a specific focus on the applicability of our theoretical results on practical models. Therefore we have a separate section (Section \ref{subsec:applications of general state convergence}) discussing the applicability of our results on many  practical situations where the exchange algorithms are used.

\section{Preliminaries}\label{sec:preliminaries}
\subsection{The Exchange Algorithm}\label{subsec:exchange algorithm}
Let $\Theta$ be the parameter space and  $\pi(\theta|x)$ be the target density on $\Theta$, the standard Metropolis--Hastings Algorithm (MHMC) is described in Algorithm \ref{alg:MHMC}.
\begin{algorithm}
	\caption{Metropolis--Hastings Algorithm (MHMC)}\label{alg:MHMC}
	\hspace*{\algorithmicindent} \textbf{Input:} initial setting $\theta$, number of iterations $T$, Markov transition kernel $q$ \\
	
	\begin{algorithmic}[1]
		
		\For{$t= 1,\cdots T$}
		\State Propose $\theta'\sim q(\theta,\theta')$
		\State Compute $$a(\theta,\theta') = \frac{q(\theta',\theta)\pi(\theta'|x)}{q(\theta,\theta')\pi(\theta|x)} $$
		\State Draw $r \sim \text{Uniform}[0,1]$
		\State \textbf{If} $(r< a)$ \textbf{then} set $\theta =\theta'$
		\EndFor
	\end{algorithmic}
\end{algorithm}

However, in our setting the posterior density  has  expression $
\pi(\theta|x) \propto \pi(\theta)  \frac{f_\theta(x)}{Z(\theta)},
$
where $Z(\theta)$ is an unknown function of $\theta$. Therefore, at each step the acceptance ratio 
\begin{equation}
\min\bigg(1, \frac{q(\theta',\theta)\pi(\theta'|x)}{q(\theta,\theta')\pi(\theta|x)}\bigg) =\min\bigg(1,\frac{q(\theta',\theta)\pi(\theta')f_{\theta'}(x)}{q(\theta,\theta')\pi(\theta)f_{\theta}(x)}\cdot \frac{Z(\theta)}{Z(\theta')}\bigg),
\end{equation}

contains an intractable term $Z(\theta)/Z(\theta')$. 

The exchange algorithm described below in Algorithm \ref{alg:exchange} is a clever extension of MHMC which uses an auxiliary variable at each step to estimate the unknown ratio of $Z(\theta)/Z(\theta')$. 
\begin{algorithm}[htbp!]
	\caption{Exchange Algorithm }\label{alg:exchange}
	\hspace*{\algorithmicindent} 
	\textbf{Input:} initial setting $\theta$, number of iterations $T$ \\
	\begin{algorithmic}[1]
		
		\For{$t= 1,\cdots T$}
		\State Generate $\theta'\sim q(\theta, \theta')$
		\State Generate an auxiliary variable $w\sim p_{\theta'}(w) = f_{\theta'}(w)/Z(\theta')$
		\State Compute $$a(\theta,\theta',w) =  \frac{\pi(\theta')q(\theta',\theta)f_{\theta'}(x)}{\pi(\theta)q(\theta,\theta')f_{\theta}(x)}\cdot \frac{f_\theta(w)}{f_{\theta'}(w)}$$
		\State Draw $r \sim \text{Uniform}[0,1]$
		\State \textbf{If} $(r< a)$ \textbf{then} set $\theta =\theta'$
		\EndFor
	\end{algorithmic}
\end{algorithm}

If we compare the exchange algorithm with the Metropolis--Hastings Algorithm (Algorithm \ref{alg:MHMC}), it turns out the only difference is the uncomputable ratio $Z(\theta)/Z(\theta')$ appeared in Algorithm  \ref{alg:MHMC} is replaced by $f_\theta(w)/f_{\theta'}(w)$ in Algorithm \ref{alg:exchange}, where $w$ is the auxiliary variable generated in each step. Roughly speaking, the exchange algorithm uses the importance sampling-type estimator $f_\theta(w)/f_{\theta'}(w)$ to estimate $Z(\theta)/Z(\theta')$  and plugs it into the uncomputable term. The exchange algorithm is easy to implement and is simple in the sense that it differs from the original Metropolis--Hastings algorithm by only an extra auxiliary variable in one step. Meanwhile, the estimator is cleverly designed so the correct stationary distribution is still preserved. 

Practitioners also use the exchange algorithm in Ising Models \cite{park2018bayesian}, Exponential Random Graph Model (ERGM) \cite{caimo2011bayesian}, spatial autoregressive (SAR) model \cite{hsieh2016social}, spatial interaction point process \cite{park2018bayesian}, Bayesian hypothesis testing \cite{diaconis2018bayesian} and so on. However, theoretical studies for doubly intractable distributions and the exchange algorithm are still very limited. 
Murray et al. proved the detailed-balance equation holds for the exchange algorithm in their original paper \cite{murray2012mcmc}. Nicholls et al. \cite{nicholls2012coupled} gave a sufficient condition for a minorization condition of the exchange chains.  Habeck et al. \cite{habeck2020stability} provided stability properties of doubly-intractable distributions. 
Medina-Aguayo et al. \cite{medina2020perturbation} provided guarantees  for the Monte Carlo within Metropolis algorithm for approximate sampling of doubly intractable distributions. Andrieu et al. \cite{andrieu2018utility} introduced a new class of MCMC algorithms, which contains the exchange algorithm as a special case, and discussed their asymptotic variance properties comparing with the original algorithms. However, it seems the only existing result concerning the convergence rate of the exchange algorithm is in \cite{nicholls2012coupled}, but  it only discussed the uniformly ergodic case, and the proposed conditions seem to be strong and are generally not satisfied in an unbounded parameter space, which is of practitioner's main interest. For example, geometric ergodicity is the usual notion of a chain having a `good' convergence rate. But there is no result showing whether the exchange algorithm is  geometrically ergodic or not. This motivates us to study the theoretical properties of the exchange algorithm.

\subsection{Markov Chain Convergence}
Let $X_0, X_1, \dotsc, X_n$ be a reversible, $\phi$-irreducible and aperiodic Markov chain with  stationary distribution $\pi$.  Let  $P$ be its transition kernel on a state space with countably generated $\sigma$ algebra. It is standard in  Markov chain theory (see, for example, Meyn and Tweedie \cite{meyn2012markov}, Chapter 13) that we have 
$
\|P^n(x,\cdot) - \pi\|_\TV \rightarrow 0
$
as $n\rightarrow \infty$ for $\pi$-a.e. $x$, where $\TV$ stands for the  total-variation distance. The reversible, irreducible and aperiodic conditions are usually easy to check and are generally satisfied in Metropolis--Hastings algorithms, so we will assume all these conditions are  satisfied in this paper henceforth. 

Furthermore, a Markov chain  is said to be \textbf{uniformly ergodic} if 
\begin{equation}\label{eqn:uniform ergodic}
\sup_x \|P^n(x,\cdot) - \pi\|_\TV \leq C r^n
\end{equation}
for $C > 0$ and $0 < r < 1$,
and  \textbf{geometrically ergodic} if there exists a finite function $C(x)$ such that
\begin{equation}\label{eqn:geometric ergodic}
\|P^n(x,\cdot) - \pi\|_\TV \leq C(x) r^n,
\end{equation}
for some $0<r<1$ and $\pi$-a.e.  $x$ .

Geometric  ergodicity plays an important role in the theory of nonasymptotic convergence estimates for MCMC algorithms, as well as the existence of the central limit theorem for Markov chains. Various conditions for the geometric ergodicity are discussed in \cite{rosenthal1995minorization}, \cite{mengersen1996rates}, \cite{roberts1996geometric}, \cite{roberts1997geometric},  \cite{jarner2000geometric} under different settings. We also summarize several criterias for geometric and uniform ergodicity in the next two theorems. The following theorems are taken from Chapter 15 and Chapter 16 of \cite{meyn2012markov}, Theorem 1.3 and 1.4 of \cite{mengersen1996rates}, and  Proposition 1 and Theorem 2 of  \cite{roberts1997geometric}:

\begin{theorem}[Uniform Ergodicity]\label{thm:uniform ergodicity}
	For a Markov chain with transition kernel $P$, state space $\bR^d$ and stationary distribution $\pi$, the followings are equivalent:
	\begin{enumerate}
		\item  A minorization condition holds for the whole space $\mathbb R^d$, i.e., there exists an integer $n_0 > 0$, $\delta > 0$, and a probability measure $\nu$ such that, for any $x\in \mathbb R^d$,
		$P^{n_0}(x,\cdot) \geq \delta \nu(\cdot)$
		\item The chain $P$ is uniformly ergodic, and
		$
		\|P^n(x,\cdot) - \pi\|_\TV \leq (1-\delta)^\floor{n/n_0}.
		$
	\end{enumerate}
\end{theorem}

\begin{theorem}[Geometric Ergodicity]\label{thm:geometrically ergodicity}
	For a reversible Markov chain with transition kernel $P$, state space $\bR^d$, and stationary distribution $\pi$, the followings are equivalent:
	\begin{enumerate}
		\item $P$ is geometrically ergodic
		\item There exists a function $V \geq 1$, finite at least for one point, and a measurable set $C$, such that for some $\lambda < 1$, $b < \infty$:
		\begin{equation}\label{eq:drift-minorization}
		PV(x) \leq \lambda V(x) + b \Indc_C(x)
		\end{equation}
		for all $x$
		\item There exists $0 < r <1$  such that $\sigma(P) \subset [-r, r]$.
		Here $\sigma(P):= \{\lambda: P-\lambda I ~\text{not invertible} \}$.  Here  $P$ is viewed as an operator on $L^2_0(\pi):=\{f: \bE_\pi(f^2) < \infty, \bE_\pi(f) = 0\}$, where $\pi$ is the stationary distribution.
	\end{enumerate}
\end{theorem}

\subsection{Notations}
Through out this paper, we will denote by $\pmh$  and $\pex$  the Markov transition kernel with respect to the original Metropolis--Hastings (MH) and the exchange algorithm respectively. As the exchange algorithm is defined in the setting of Bayesian inference, both chains are defined on the parameter space $\Theta$, which is considered to be $\bR^d$ or a subset of $\bR^d$ equipped with the Euclidean norm $\lVert\cdot \rVert$. We   denote by $q$ the proposal density of both $\pmh$ and $\pex$. For each $\theta$, both $\pex(\theta, \cdot)$ and $\pmh(\theta, \cdot)$ can be represented by a mixture of a continuous density and a point mass at $\theta$. We therefore denote by $\dpmh$ and $\dpex$ the continuous density part of $\pmh$ and $\pex$, respectively.  For each $\theta \in \Theta$, there is an associated probability measure with  density  (or probability mass function) $p_\theta(x) = \frac {f_\theta(x)}{Z(\theta)}$ defined on the sample space $\mathcal X$. The sample space $\mathcal X$ can  either be discrete or continuous.
\section{Theoretical results }\label{sec:theoretical exchange algorithm}
Now we are ready to discuss the theoretical properties of the exchange algorithm. We will first study its asymptotic variance, and then study its convergence rate. As the exchange algorithm $\pex$ is based on the original algorithm $\pmh$, many theoretical results here are comparison-type results. The rest of this section is organized as follows. Section \ref{subsec:asymptotic variance} discusses the asymptotic variance of the exchange algorithm. Section \ref{subsec: main results} -- \ref{subsec:Exchange chain convergence} discusses the convergence rate properties of the exchange algorithm. Section \ref{subsec:applications of general state convergence} connects our theoretical results with many practical models where the exchange algorithm is used.
\subsection{Asymptotic variance results and Peskun's ordering}\label{subsec:asymptotic variance}
We start by  proving the following simple but useful lemma, indicating that the exchange chain is always less statistically efficient comparing with the original MH chain:
\begin{lemma}[also proved in \cite{andrieu2018utility}, \cite{nicholls2012coupled}, \cite{andrieu2016establishing}]\label{lem:Peskun_order}
	
	For any $\theta, \theta' \in \Theta$, if $\theta \neq \theta'$, then the continuous density part of $\pmh$ and $\pex$ follows
	\[
	\dpex(\theta, \theta') \leq \dpmh(\theta, \theta').
	\]
\end{lemma}
This lemma shows that, the exchange algorithm is uniformly less likely to make a move compared with the original MH algorithm.

\begin{proof}
	By Jensen's inequality (the function $\min\{1,x\}$ is concave).
	\[
	\dpex(\theta, \theta') = 
	\bE_{w\sim p_{\theta'}} \min\{1, a(\theta, \theta', w)\}\leq \min\{1, \bE_{w\sim p_\theta'}a(\theta, \theta', w) \},
	\]
	where 
	\[
	a(\theta,\theta',w) = \frac{\pi(\theta')q(\theta',\theta)f_{\theta'}(x)}{\pi(\theta)q(\theta,\theta')f_{\theta}(x)}\cdot \frac{f_\theta(w)}{f_{\theta'}(w)}
	\]
	is the randomized acceptance ratio defined in Algorithm \ref{alg:exchange}.
	Meanwhile,
	\begin{align*}
	\bE_{w\sim p_{\theta'}}a(\theta, \theta', w)  &= \int \frac{\pi(\theta')q(\theta',\theta)f_{\theta'}(x)}{\pi(\theta)q(\theta,\theta')f_{\theta}(x)}\cdot \frac{f_\theta(w)}{f_{\theta'}(w)} p_{\theta'}(w) dw \\
	& = \frac{\pi(\theta')q(\theta',\theta)p_{\theta'}(x)}{\pi(\theta)q(\theta,\theta')p_{\theta}(x)}\cdot \int \frac{f_\theta(w)}{f_{\theta'}(w)} \frac{Z(\theta')}{Z(\theta)} p_{\theta'}(w) dw \\ 
	& = \frac{\pi(\theta')q(\theta',\theta)p_{\theta'}(x)}{\pi(\theta)q(\theta,\theta')p_{\theta}(x)}\cdot \int p_\theta(w) dw \\
	& = \frac{\pi(\theta'|x)q(\theta',\theta)}{\pi(\theta|x)q(\theta,\theta')} = a(\theta,\theta').\\
	\end{align*}
	Therefore
	\[
	\dpex(\theta, \theta')\leq \min\{1, a(\theta,\theta') \} = \dpmh(\theta,\theta'),
	\]
	as desired.
\end{proof}
As the two chains have the same stationary distribution, this shows $\pmh \succ \pex$ in Peskun's ordering \cite{peskun1973optimum}. Then it follows directly from \cite{peskun1973optimum} and \cite{tierney1998note}:  
\begin{theorem}[[also proved in \cite{andrieu2018utility}, \cite{nicholls2012coupled}, \cite{andrieu2016establishing}]\label{thm:asymptotic variance}
	Let $L_0^2(\pi)$ be the set of all the $L^2(\pi)$-integrable random variable with mean $0$. Define 
	
	\[\sigma^2(P,f) = \lim_{n\rightarrow \infty} \frac 1n \var_P{\sum_{i=1}^nf(X_i)}, \]
	where $X_0, X_1, \dotsc, X_n$ is a Markov chain with initial distribution $\pi$ and transition kernel $P$. Then 
	$$\sigma^2(\pmh,f) \leq \sigma^2(\pex,f)$$ 
	for all $f\in L_0^2(\pi)$.
\end{theorem}

The quantity $\sigma^2(P,f) $ is often referred to as the `asymptotic variance'. Theorem \ref{thm:asymptotic variance} proves the original  $\pmh$ chain has smaller asymptotic variance and is thus statistically more efficient than $\pex$ chain.  

\begin{remark}
	The asymptotic variance defined in Theorem \ref{thm:asymptotic variance} may be infinite. It is worth mentioning that a Markov chain is called `variance bounding' if its asymptotic variance is finite for all $f\in L_0^2(\pi)$. The relationship between variance bounding and geometric ergodicity is discussed thoroughly in \cite{roberts2008variance}.
\end{remark}

Theorem \ref{thm:asymptotic variance} is not very surprising because in each iteration of the exchange algorithm, the ratio $\frac{f_\theta(w)}{f_{\theta'}(w)}$ can be viewed as an estimator  for the unknown quantity $\frac{Z(\theta)}{Z(\theta')}$. On the other hand, the standard MH chain uses $\frac{Z(\theta)}{Z(\theta')}$ directly which can be viewed as an estimator with variance $0$. Therefore it is not surprising that the original MH chain has a smaller asymptotic variance.

However, asymptotic variance is only one measurement to evaluate a Markov chain. Another natural way of evaluating a Markov chain is the speed of convergence to stationary distribution. Even though $\pex$ is dominated by $\pmh$ in Peskun's order, the following simple example shows it is possible that $\pex$ converges to stationary distribution uniformly faster than $\pmh$.
\begin{example}[Two point example]\label{eg:two point bayes}
	Let $X\sim \Bern(\theta)$, where the parameter space $\Theta$ only contains two points: $\Theta = \bigg\{\theta_1 = \frac 14, \theta_2 = \frac 34\bigg\}$. 
	Suppose the observed data is only one single point $x = 1$. Suppose the prior measure on $\Theta$ is defined by $\pi(\theta_1) = \frac 34, \pi(\theta_2) = \frac 14$. It is not hard to compute the posterior measure:
	$
	\pi(\theta_1|x) = \pi(\theta_2|x) = \frac 12,
	$
	which is a uniform measure on $\Theta$.
	We further assume the transition matrix equals
	$
	\begin{pmatrix} 
	0~~~&1\\
	1~~~&0\\
	\end{pmatrix}.
	$
	
	It is clear that all the moves of the Metropolis--Hastings chain will be accepted, hence $\pmh$ has transition matrix: $
	\begin{pmatrix} 
	0~~~&1\\
	1~~~&0\\
	\end{pmatrix}.
	$
	
	On the other hand, the transition function for $\pex$ chain can be computed by:
	\[
	\pex(\theta_1,\theta_2) =   \Prob_{\theta_2}(w = 0) + \frac 13  \Prob_{\theta_2}(w = 1) = 0.5.
	\]
	Similarly $\pex(\theta_2,\theta_1) = \frac 12,$ so the transition matrix would be 
	$
	\begin{pmatrix} 
	0.5~~~&0.5\\
	0.5~~~&0.5\\
	\end{pmatrix}.
	$
	Therefore, with any initialization, $\pex$ converges to the stationary distribution after one step. However, $\pmh$ never converges as it jumps back and forth between $\theta_1$ and $\theta_2$.
\end{example}

\begin{remark}
	In Example \ref{eg:two point bayes} above, $\pmh$ never converges because it is a periodic chain, i.e., its smallest eigenvalue equals $-1$. But even if we assume  $\pmh$ is aperiodic, it is still possible  that $\pmh$ coverges slower than $\pex$, as we could tilt  the transition matrix above a little bit, for example, let the transision matrix for $\pmh$ be
	$
	\begin{pmatrix} 
	\epsilon& 1-\epsilon\\
	1-\epsilon&\epsilon\\
	\end{pmatrix},
	$
	then $\pex$ chain still converges faster.
\end{remark}

Combining Theorem \ref{thm:asymptotic variance} and Example \ref{eg:two point bayes}, we can conclude that
\begin{enumerate}
	\item In terms of asymptotic variance, $\pmh$ chain is at least as good as  $\pex$,
	\item In terms of distributional convergence, it is not possible to derive a general ordering between $\pmh$ and $\pex$ chain.
\end{enumerate}

The above results tell us the exchange chain might converge faster or more slowly than the original MH chain. In the rest part of this section, we will further investigate the convergence speed of $\pex$.
\subsection{Convergence rate: summary of main results}\label{subsec: main results}
In this part, we  study the convergence properties of the exchange algorithm.  The exchange algorithm can be viewed as a variant of the idealised but impractical  Metropolis--Hastings algorithm. Therefore, suppose one knows the convergence speed of one algorithm, it is natural  to ask if the other algorithm also has a similar convergence speed. For example, one can ask questions like:
\begin{itemize}
	\item \textbf{Question 1:} Suppose $\pex$ is geometrically ergodic, is the original chain  also geometrically ergodic?
\end{itemize}

Or the reverse
\begin{itemize}
	\item \textbf{Question 2:} Suppose $\pmh$ is geometrically ergodic, is the exchange chain also geometrically ergodic? If not, can we find sufficient conditions to ensure $\pex$ `inherits' the geometric ergodicity of the original chain?
\end{itemize}

We will answer both of the two questions in the rest of this section. The second question is probably more interesting in a practitioner's point of view. In real settings, usually we can study the convergence rate of $\pmh$, though it is not practically implementable. Therefore theoretical guarantees of $\pex$ would justify the usefulness of the exchange algorithm.

Before everything is rigorously stated, we state our results in a heuristic way here. All the results mentioned below are formally stated and proved in Section \ref{subsec:Exchange chain convergence}.
\begin{itemize}
	\item (Question 1) If $\pex$ is geometrically ergodic, there is no guarantee that $\pmh$ is also geometrically ergodic. In fact, Example \ref{eg:two point bayes} gives  such a counterexample.
	\item (Question 1) If $\pex$ is geometrically ergodic, then any `lazy' version (defined later) of $\pmh$ must also be geometrically ergodic (Theorem \ref{thm:lazy MH chain}). If the original chain $\pmh$ satisfies some further conditions (see Corollary \ref{cor:pmh with lower bounded rejection} for details), then the geometric ergodicity of $\pex$ implies the geometrically ergodicity of $\pmh$. 
	\item (Question 2) If $\pmh$  is geometrically ergodic, we have an  example to show the exchange algorithm is not necessarily geometrically ergodic. 
	\item (Question 2) If  $\pmh$  is geometrically ergodic, we have established sufficient conditions to ensure the geometric ergodicity of the exchange algorithm (Theorem \ref{thm:geometric ergodicity inherit, uniform case}, \ref{thm:Geometric ergodic inhirit}). Theorem \ref{thm:geometric ergodicity inherit, uniform case} gives a general condition for `geometric ergodicity inheritance' without further assumption on the structure of transition kernels.  Theorem \ref{thm:Geometric ergodic inhirit} gives a much weaker condition but with the further assumption that the proposal kernel is symmetric.
\end{itemize}

\subsection{Exchange chain convergence}\label{subsec:Exchange chain convergence}

Now we are ready to study the convergence properties for $\pex$. 
Though Example \ref{eg:two point bayes} gives us a negative example, indicating the exchange chain may converge  faster than the original chain. The next theorem shows, after making the original chain `lazy', the original chain is no worse than the exchange chain, which answers Question 1 completely. 
\begin{theorem}\label{thm:lazy MH chain}
	
	Suppose the exchange chain $\pex$ is uniformly/geometrically ergodic, then for any $ 0 < \lambda < 1$, the chain $\pmh(\lambda)$ is also uniformly/geometrically ergodic. Here $\pmh(\lambda)$  is the lazy version of $\pmh$, defined by 
	\[
	\pmh(\lambda) := \lambda \pmh + (1-\lambda) I.
	\]
\end{theorem}
\begin{proof}
	First suppose $\pex$ is uniformly ergodic, then Theorem \ref{thm:uniform ergodicity} shows there exists here exists an integer $n_0 > 0$, $\delta > 0$, and a probability measure $\nu$ such that, for any $x\in \mathbb R^p$,
	\[\pex^{n_0}(x,\cdot) \geq \delta \nu(\cdot)\].
	On the other hand, for any measurable set $B$ and any point $x$, if $x\in B$, then we have 
	\begin{align*}
	\pmh(\lambda)(x,B) &= \pmh(\lambda)(x,\{x\}) +  \pmh(\lambda)(x,B/\{x\}) \\
	&\geq (1-\lambda) + \lambda \pex(x,B/\{x\}) \\
	&\ \geq \lambda_0 \pex(x,B), 
	\end{align*}
	where $\lambda_0 = \min\{\lambda, 1-\lambda\}$. The same result holds if $x\notin B$. Therefore, $$\pmh^{n_0}(\lambda)(x,\cdot) \geq \lambda_0^{n_0} \pex(x,\cdot) \geq  \lambda_0^{n_0}\delta \nu(\cdot) $$ for all $x$. Thus $\pmh(\lambda)$ is uniformly ergodic by Theorem \ref{thm:uniform ergodicity}.
	
	Now suppose $\pex$ is geometrically ergodic. Define $\sigma(\pex)$, the spectrum of $\pex$, be the set of real numbers $\lambda$ such that  the operator $\pex - \lambda I$ is not invertible on $L_0^2(\pi)$, and define $\sigma(\pmh), \sigma(\pmh(\lambda))$ accordingly. 	By Theorem \ref{thm:geometrically ergodicity}, we have $\sigma(\pex)\subset [-r, r]$ for some $r < 1$. 
	
	Let $m(\lambda):= \inf\sigma(\pmh(\lambda))$ and $M(\lambda):= \sup\sigma(\pmh(\lambda))$. In particular let $m = m(1)$ and $M = M(1)$. Then we have
	\[
	m(\lambda) = \lambda m + (1-\lambda) \geq 1- 2\lambda > -1.
	\]
	
	Meanwhile, as the exchange chain is dominated by the original MH chain in Peskun's ordering, it is shown by Tierney \cite{tierney1998note} Lemma 3 that the supremum of $\pex$'s spectrum should be no less than the supremum of $\pmh$'s spectrum. That is, $M < \sup\sigma(\pex) \leq r$. 
	
	Therefore $$M(\lambda)  = \lambda M + (1-\lambda) \leq  \lambda r + (1-\lambda) < 1.$$ Since $-1 < m(\lambda) \leq M(\lambda) < 1,$ there exists $r(\lambda)$ such that $$\sigma(\pmh(\lambda))\subset [-r(\lambda), r(\lambda)].$$ The lazy version of the MH algorithm $\pmh(\lambda)$ is thus geometrically ergodic, as desired. 
\end{proof}

In fact, the `laziness modification' is often not necessary when the original chain $\pmh$ has spectrum strictly bounded above from $-1$. The next result follows immediately from Theorem \ref{thm:lazy MH chain}.
\begin{corollary}\label{cor:pmh with larger than -1  spectrum }
	Suppose the exchange chain $\pex$ is geometrically ergodic, and the spectrum of the MH chain  $\pmh$ belongs to the interval $[-r, 1]$ for some $r<1$, then $\pmh$ is also geometrically ergodic.
\end{corollary}
\begin{proof}
	It suffices to show the supremum of $\sigma(\pmh)$ is strictly less than $1$, which follows immediately from our assumption that the supremum of $\sigma(\pex)$ is strictly less than $1$ and the Peskun's ordering result proved in Lemma \ref{lem:Peskun_order}.
\end{proof}
There are many Metropolis--Hastings algorithms that are known to have spectrum in the set $[-r,1]$, or even $[0,1]$. The following corollary provides  a few examples,  see also Rudolf and Ullrich \cite{rudolf2013positivity} for further discussions.
\begin{corollary}\label{cor:pmh with lower bounded rejection}
	If the original Metropolis--Hastings algorithm $\pmh$ satisfies any of the following conditions, then its spectrum   belongs to the interval $[-r, 1]$ for some $r<1$. In particular, if $\pmh$ satisfies any of the Contiditons 2--5. Then the $\pmh$ is a positive operator, i.e., its spectrum is a subset of the interval $[0,1]$.  Therefore,  it follows from Theorem \ref{thm:lazy MH chain} that the geometric ergodicity of $\pex$ implies the geometricity ergodicity of $\pmh$.
	\begin{enumerate}
		\item 	If the $\pmh$ chain has rejection probability (i.e. $\pmh(\theta, \{\theta\})$) uniformly bounded below from $0$, in other words, there exists $c > 0$ such that:
		\[
		\pmh(\theta, \{\theta\}) > c
		\]
		uniformly over $\theta \in \Theta$
		\item If the $\pmh$ algorithm is an independent Metropolis--Hastings algorithm (IMH)
		\item If the $\pmh$ algorithm is a random-scan Gibbs sampler
		\item If the $\pmh$ algorithm is a hit-and-run  sampler
		\item If the $\pmh$ algorithm has proposal distribution of the following form:
		\[
		q(\theta, \theta') = \int q(\theta, \eta) q(\theta', \eta) d\eta .
		\]
	\end{enumerate}
	
\end{corollary}

\begin{proof}
	The proof when $\pmh$ satisfies any of Conditions 2--4 can be found in Rudolf and Ullrich \cite{rudolf2013positivity}. The proof when $\pmh$ satisfies Condition 5 can be found in Andrieu and Vihola \cite{andrieu2015convergence}, Proposition 16. Therefore we only prove the case when $\pmh$ satisfies condition 1. 
	
	Since we have
	\[
	\pmh(\theta, \{\theta\}) > c
	\]
	uniformly over $\theta \in \Theta$. It allows us to write $\pmh$ as a convex combination of two Markov operators, i.e., 
	\[
	\pmh = c I + (1-c)\tilde P_{\text{MH}},
	\]
	where $\tilde P_{\text{MH}} := \frac{\pmh - cI }{1-c}$ is a well-defined Markov transition kernel. Since the spectrum of the identity operator is precisely $\{1\}$, we have:
	\[
	\inf \sigma(\pmh) = c + (1-c) \inf\sigma(\tilde P_{\text{MH}})	\geq c + (1-c)\times(-1) = 2c -1 > -1,\]
	as desired.
\end{proof}

Theorem \ref{thm:lazy MH chain},  Corollary \ref{cor:pmh with larger than -1  spectrum }, and Corollary \ref{cor:pmh with lower bounded rejection} essentially show  the (at most slightly modified) original chain is geometrically ergodic when the exchange algorithm is geometrically ergodic.  In other words, the geometric ergodicity of $\pmh$ is essentially the \textbf{necessary condition} for  the geometric ergodicity of $\pex$. However, the next example shows this condition is not sufficient. 

\begin{example}[Exponential likelihood with Gamma prior]\label{eg: exponential gamma}
	
	Suppose the likelihood is $\mathrm{Exp}(\theta)$, that is, $$p_\theta(x) = \theta e^{-\theta x},$$ prior distribution is $\mathrm{Exp}(1)$. Under this setting, it is easy to compute the posterior distribution:
	\[
	\pi(\theta|x) \propto \theta e^{-\theta(x+1)}
	\]
	which is a $\text{Gamma}(2,x+1)$ distribution. Consider an independence Metropolis--Hastings sampler with an $\text{Gamma}(2,x+1)$ proposal, that is,
	\[q(\theta, \theta') = \pi(\theta'|x).\]
	As the proposal is precisely the posterior, this independence Metropolis--Hastings chain will converge perfectly after one step, and therefore is obviously geometrically ergodic. On the other hand, the continuous density part  of $\pex$ when $\theta\neq \theta'$ is:
	\begin{align*}
	\dpex(\theta,\theta')  &= q(\theta,\theta') \bE_{p_{\theta'}}{\min\{\frac{\pi(\theta'|x)q(\theta',\theta)}{\pi(\theta|x)q(\theta,\theta')}\cdot \frac{p_\theta(w)}{p_{\theta'}(w)},1\}} \\
	& = \pi(\theta'|x)\bE_{p_\theta'}{\min\{ \frac{p_\theta(w)}{p_{\theta'}(w)},1\}} \\
	& = \pi(\theta'|x) \int \min\{\theta e^{-\theta w}, \theta'e^{-\theta' w}\} dw \\
	& = \pi(\theta'|x) ( 1 - e^{-w_0\min\{\theta, \theta'\}} + e^{-w_0(\theta,\theta')\max\{\theta, \theta'\}}),
	\end{align*}
	where $$w_0(\theta,\theta') = \frac{\log(\theta) - \log(\theta')}{\theta - \theta'}$$ is the only solution for equation:
	\[
	\theta e^{-\theta w} = \theta'e^{-\theta' w}
	\]
	for any fixed $\theta\neq \theta'$.
	
	We can also compute the `rejection probability' at point $\theta$: 
	\begin{align*}
	\pex(\theta,\{\theta\})  &= 1 - \int \dpex(\theta,\theta') d\theta'\\
	& = \int \pi(\theta'|x) (e^{-w_0(\theta,\theta')\min\{\theta, \theta'\}} - e^{-w_0(\theta,\theta')\max\{\theta, \theta'\}}) d\theta'\\
	& = \int_0^\theta  \pi(\theta'|x) (e^{-w_0(\theta,\theta') \theta'} - e^{-w_0(\theta,\theta')\theta}) d\theta' + \int_{\theta}^\infty  \pi(\theta'|x)  (e^{-w_0(\theta,\theta')\theta} - e^{-w_0(\theta,\theta')\theta'}) d\theta'\\
	\end{align*}
	
	Notice that for each fixed $\theta'$, when $\theta$ goes to infinity, we have $-w_0(\theta,\theta')\theta\rightarrow -\infty$ and $-w_0(\theta,\theta')\theta' \rightarrow 0.$  Therefore the first term of the integration goes to $$\int_0^\infty  \pi(\theta'|x) d\theta' = 1 $$ as $\theta \rightarrow \infty$ by Lebesgue's dominated convergence theorem (with control function $\pi(\theta'|x)$). The second term goes to $0$ as  $\theta \rightarrow \infty$. 
	
	Therefore, we have:
	\[
	\mathrm{esssup}_\theta \pex(\theta,\{\theta\}) = 1.
	\]
	It is proved by Roberts and Tweedie \cite{roberts1996geometric} (Thm 5.1) that if a MH chain is geometrically ergodic, then its rejection probability is necessarily bounded away from unity. Therefore, the exchange algorithm in this example is not geometric ergodic. 
\end{example}

Theorem \ref{thm:lazy MH chain} provides a necessary condition for the geometric ergodicity of the exchange algorithm. In practice, however, practitioners are more interested in the `reverse problem'. When the original algorithm $\pmh$ is geometrically ergodic (though not implementable in practice), it is of practitioner's main interest to establish the sufficient conditions for the geometrically ergodicity of the exchange algorithm. We will focus on the sufficient conditions in the next two subsections.

\subsubsection{Geometric ergodicity of the exchange algorithm when the likelihood ratio is uniformly non-negligible}

Example \ref{eg: exponential gamma} answers half of Question 2. That is, if $\pmh$ is geometrically ergodic, the exchange algorithm is not necessarily geometrically ergodic.  Now we focus on establishing sufficient conditions such that $\pex$ will `inherit' the convergence rate of $\pmh$. We start with the following lemma, which is slightly different from Corollary 11 in Roberts and Rosenthal \cite{roberts2008variance}.
\begin{lemma}\label{lem:uniform likelihood ratio on spectrum}
	Let $P_1$, $P_2$ be two reversible Markov transition kernels with the same stationary distribution $\pi$. If there exists $\epsilon > 0$ such that $P_1(x,M) \geq \epsilon P_2(x,M)$ for any $x$ and any measureable set $M$, and $P_2$ is geometrically ergodic, then $P_1$ is geometrically ergodic.
\end{lemma}

\begin{proof}
	We view both $P_1$ and $P_2$ as self-adjoint operators on $L^2_0(\pi)$. By assumption, we can write $P_1$ as 
	$
	P_1 = \epsilon P_2 + (1-\epsilon) P_{\mathrm{res}},
	$
	where $P_{\text{res}} =  \frac{P_1 - \epsilon P_2}{1-\epsilon}$ is a valid, self-adjoint Markov operator.
	
	Let $\sigma(P_1)$, the spectrum of $P_1$, be the set of all the complex numbers $\lambda$ such that $P_1 - \lambda I$ is not invertible, and define $\sigma(P_2), \sigma(P_{\text{res}})$ in the same way. It follows from the spectral theory of self-adjoint operators that the spectrum of both $P_1$ and $P_2$  are real. 
	Let $r(P_1) := \sup \{|\lambda|: \lambda \in \sigma(P_1) \} $ be the spectral radius of $P_1$, and define $r(P_2), r(P_{\mathrm{res}})$
	in the same way. Since $P_1, P_2, P_{\mathrm{res}}$ are all self-adjoint, the spectral radius of each operator coincides with their operator norms (see Proposition 9.2 of \cite{hunter2001applied} for a proof). Moreover, Theorem $2$ in \cite{roberts1997geometric} says a reversible Markov operator $P$ is geometrically ergodic if and only if $r(P)\leq r$ for $ 0< r < 1$.  Thus we can find $0< r_2<1$ such that $r(P_2) = \lVert P_2\rVert\leq r_2$.

Now we study the operator norm of $P_1$,  we have
  $$
  \lVert P_1 \rVert \leq \epsilon \lVert P_2 \rVert + (1-\epsilon) \lVert  P_{\mathrm{res}} \rVert \leq \epsilon r_2 + (1-\epsilon) < 1
  $$
  where the first inequality follows from the triangle inequality and the second inequality follows from the fact that $\lVert P_2 \rVert \leq r_2$  and $\lVert P_{\mathrm{res}} \rVert \leq 1 $ (since $P_{\mathrm{res}}$ is a Markov operator). Let $r_1 :=  r_2 + (1-\epsilon) < 1$, we now have $r(P_1) = \lVert P_1 \rVert \leq r_1$, thus we conclude $P_1$ is also geometrically ergodic. 
\end{proof}

To connect Lemma \ref{lem:uniform likelihood ratio on spectrum} with the geometric ergodicity of $\pex$. We need the following definition.
\begin{definition}[Uniformly Non-negligible Likelihood Ratio]\label{def: non-negligible}
	With all the notations as above, let $A_{\theta,\theta'}(s)$ be the set 
	\[
	A_{\theta,\theta'}(s) := \{x\in \mathcal X: p_\theta(x)>sp_{\theta'}(x)\}.
	\]
	The likelihood ratio function  is called \textbf{uniformly non-negligible} if there exist  $\epsilon > 0$ and $\delta > 0$, such that $
	\Prob_{\theta'}(A_{\theta,\theta'}(\delta)) > \epsilon$
	uniformly over $(\theta, \theta') \in \Theta\times \Theta$.
\end{definition}
The likelihood ratio function is uniformly non-negligible if the set $A_{\theta,\theta'}(\delta)$ has a uniformly positive probability. Our next result shows that, when the likelihood ratio function is non-negligible, then $\pex$ can `inherit' the geometric ergodicity from $\pmh$. 
\begin{theorem}\label{thm:geometric ergodicity inherit, uniform case}
	Suppose  the likelihood ratio function is uniformly non-negligible, and   $\pmh$ is uniformly/geometrically ergodic.  Then the exchange chain is also uniformly/geometrically ergodic, respectively.
\end{theorem}
\begin{proof}
	For every $\theta$, recall that $\pex(\theta,\cdot)$ can be represented as a mixture of a continuous density $\dpex$ and a point mass at $\theta$, the continuous density part follows
	\begin{align*}
	\dpex(\theta,\theta')  &= q(\theta,\theta') \bE_{p_{\theta'}}{\min\{\frac{\pi(\theta')q(\theta',\theta)f_{\theta'}(x)}{\pi(\theta)q(\theta,\theta')f_{\theta}(x)}\cdot \frac{f_\theta(w)}{f_{\theta'}(w)},1\}} \\
	& = q(\theta,\theta')\bE_{p_{\theta'}}{\min\{\frac{\pi(\theta'|x)q(\theta',\theta)}{\pi(\theta|x)q(\theta,\theta')}\cdot \frac{p_\theta(w)}{p_{\theta'}(w)},1\}} \\
	& \geq \delta \Prob_{\theta'}(A_{\theta,\theta'}(\delta))q(\theta,\theta')\min\{\frac{\pi(\theta'|x)q(\theta',\theta)}{\pi(\theta|x)q(\theta,\theta')},1\} \\
	& = \delta \Prob_{\theta'}(A_{\theta,\theta'}(\delta)) \dpmh(\theta,\theta')\\
	& \geq \delta \epsilon \dpmh(\theta,\theta').
	\end{align*}
	Lemma \ref{lem:Peskun_order} proves $\dpex(\theta,\{\theta\}) \geq \dpmh(\theta,\{\theta\})$ for all $\theta\in \Theta$, therefore we have $\pex(\theta, M) \geq \delta\epsilon \pmh(\theta,M)$ for any $\theta$ and any measurable set $M$. If $\pmh$ is geometrically ergodic, it follows directly from Lemma \ref{lem:uniform likelihood ratio on spectrum} that $\pex$ is geometrically ergodic.  	If $\pmh$ is uniformly ergodic, Theorem \ref{thm:uniform ergodicity} proves there exists $n_0$ and $\nu$ such that $\pmh^{n_0}(\theta, \cdot)\geq \nu(\cdot)$ for any $\theta\in \Theta$. On the other hand, we have:
	\[
	\pex^{n_0}(\theta,\cdot) \geq (\epsilon\delta)^{n_0} \pmh^{n_0}(\theta,\cdot)\geq (\epsilon\delta)^{n_0} \nu(\cdot),
	\]
	so $\pex$ also satisfies the minorization condition on the whole space.  It then follows from Theorem \ref{thm:uniform ergodicity} that $\pex$ is uniformly ergodic. 
\end{proof}
Theorem \ref{thm:geometric ergodicity inherit, uniform case} can be applied to almost all the cases where the parameter space is bounded or compact. We provide one example here:
\begin{example}[Beta-Binomial model]\label{eg:Beta-Binomial 1}
	Consider the following Beta-Binomial example. Let 
	\begin{align*}
	p_\theta(x) = \binom{n}{x} \theta^x (1-\theta)^{n-x}
	\end{align*}
	be the Binomial distribution with parameter $\theta\in\Theta = [\theta_1,\theta_2]$ where $0 < \theta_1 < \theta_2 < 1$, let 
	\[
	\pi(\theta) \propto \theta^{a-1} (1-\theta)^{b-1} \mathbb I(\theta \in [\theta_1, \theta_2])
	\]
	be a truncated Beta prior on $[\theta_1,\theta_2]$, where $a,b$ are prefixed positive numbers. Given data $x$, we would like to sample from the posterior distribution $\pi(\theta|x)$ with density
	\[
	\pi(\theta|x) \propto \theta^{a+x-1} (1-\theta)^{n+b-x-1} \mathbb I(\theta \in [\theta_1, \theta_2])\] which is a truncated $\mathrm{Beta} (x+a,n-x+b)$ distributed random variable.
	
	Consider  an  independence  Metropolis--Hastings sampler with a $\mathrm{Unif}[\theta_1,\theta_2]$ proposal distribution, then the continuous part of $\pmh$ follows:
	\[
	\dpmh(\theta, \theta')=
	\min\bigg\{1,  \frac{(\theta')^{a+x-1}(1-\theta')^{n+b-x-1}}{\theta^{a+x-1} (1-\theta)^{n+b-x-1}}\bigg\}.
	\]

Let $\theta^\star$ be the maximizer for $\theta^{a+x-1} (1-\theta)^{n+b-x-1}$, for $\theta\neq \theta'$, we have:
\[
\dpmh(\theta,\theta') \geq  \frac{C(\theta_1,\theta_2, a, b)} {(1-\theta^\star)^{n+b-x-1}}\pi(\theta'|x),
\]
where  $C(\theta_1,\theta_2,a,b)$ is the normalizing constant for the truncated Beta distribution.
 Therefore $\dpmh$ is uniformly ergodic by Theorem \ref{thm:uniform ergodicity}. Furthermore,  fix $\delta \in (0,1)$,  given any $\theta, \theta'\in [\theta_1, \theta_2]$, the set $A_{\theta,\theta'}(\delta)$ has a positive probability under $p_{\theta'}$. Since $\theta,\theta'$ only take values from a compact set $[\theta_1,\theta_2]\times [\theta_1,\theta_2]$, there exists a constant $\epsilon>0$ such that 
\[
\Prob_{\theta'}(A_{\theta,\theta'}(\delta)) > \epsilon
\]
uniformly. Hence the exchange chain is uniformly ergodic by Theorem \ref{thm:geometric ergodicity inherit, uniform case}.
\end{example}

However, the uniform probability condition for $A_{\theta,\theta'}(\delta)$ is usually too strong when the parameter space is unbounded. Note that 
\[
1 - \lVert p_\theta - p_{\theta'} \rVert_\TV = \int \min\{p_\theta(x), p_{\theta'}(x) \} dx = \int  \min\{\frac{p_\theta(x)}{p_{\theta'}(x)},1 \} p_{\theta'}(x) dx \geq \delta \Prob_{\theta'}(A_{\theta,\theta'}(\delta)).
\]

Therefore the condition $\Prob_{\theta'}(A_{\theta,\theta'}(\delta)) > \epsilon$ directly implies $\lVert p_\theta - p_{\theta'} \rVert_\TV \leq 1-\epsilon\delta$ uniformly over $\theta,\theta'$. However, when the parameter space is $\Theta$ is unbounded (say $\bR^n$), most   practical models will have $p_\theta$ and $p_{\theta'}$ far away from each other, i.e, 
\[
\mathrm{esssup}_{\theta,\theta'}\lVert p_\theta - p_{\theta'} \rVert_\TV = 1.
\]
Thus Theorem \ref{thm:geometric ergodicity inherit, uniform case} can not be directly applied in these cases. The next part gives weaker sufficient conditions which can be applied in unbounded parameter space.

\subsubsection{Geometric ergodicity of the exchange algorithm with a random-walk proposal }

To establish sufficient conditions on unbounded parameter space, we will first need the following assumptions on the $\pmh$, the conditions are similar to \cite{mengersen1996rates}\cite{roberts1996geometric}\cite{jarner2000geometric}  and are usually reasonable in practical settings. It is known \cite{roberts1996geometric}\cite{johnson2012variable} that the sufficient conditions for the geometric ergodicity of multi-dimensional random-walk Metropolis--Hastings algorithms are slightly stronger than conditions for one-dimensional algorithms. Therefore we will discuss the  cases where the dimensionality of the state space equals one and is greater than one separately. 

When  the  state space is $\bR^d$ with $d > 1$, we will say that a Metropolis--Hastings chain $\pmh$ satisfies assumption ($\mathcal A$) if it:

\begin{itemize}
	\itemaone  \label{ass:expo tail} has a target density $\pi$ which is positive and has continuous first derivatives such that:
	\begin{enumerate}
		\item $\lim\limits_{\lVert \theta \rVert \rightarrow \infty} \frac{\theta}{\lVert \theta \rVert}\cdot\nabla \log\pi(\theta) = -\infty$
		\item $\limsup\limits_{\lVert \theta \rVert \rightarrow \infty} \frac{\theta}{\lVert \theta \rVert}\cdot\frac{\nabla \log\pi(\theta)}{\lVert \nabla \log\pi(\theta) \rVert} < 0$
    \end{enumerate}
      \itematwo \label{ass: random walk}has a random-walk proposal density $q$, that is, $q(\theta, \theta') = q(\theta',\theta) = q(\lVert\theta-\theta'\rVert) $. Furthermore, $q$ is bounded away from $0$ in a neigborhood of the origin, which means there exists some $\delta_q > 0$ and $\epsilon_q >0$ such that $q(s_1, s_2)\geq \epsilon_q$ if $\lVert s_1 - s_2\rVert\leq \delta_q$.

    %  has a target density $\pi$ which is log-concave in the tails, that is, there exists some  constant $\alpha > 0$ and $x_1 > 0$ such that, for all $y > x > x_1$:
	%\[
	%\log \pi(x) - \log \pi(y)  \geq \alpha (y - x),
	%\]
	%and for all $ y < x < -x_1$:
	%\[
	%\log \pi(x) - \log \pi(y)  \geq \alpha (x - y).
	%\]
\end{itemize}

Assumption $(\mathcal A)$ is the condition in Theorem 4.3 of Jarner and Hanson \cite{jarner2000geometric}. It is also a generalization of the results in Roberts and Tweedie \cite{roberts1996geometric}. It is shown that any random-walk Metropolis--Hastings chain satisfying $(\mathcal A)$ is geometrically ergodic.  Assumption ($\mathcal A2$)  requires a random-walk proposal kernel. The first half of  ($\mathcal A1$)  requires the tail of  the target density decays super-exponentially when $\lVert \theta \rVert$ goes to infinity. The second half of ($\mathcal A1$) is a  curvature condition which requires the contour manifolds of the target density is non-degenerate in the tails. It is shown in Theorem 4.4 of  \cite{jarner2000geometric} that assumption  $(\mathcal A)$ is stable under translation, rotation, positive linear combination, and pointwise multiplication. Densities satisfying $(\mathcal A)$ includes multivariate Gaussian, mixture of multivariate Gaussians, and densities of the form $\pi(\theta)\sim h(\theta)e^{-p(\theta)}$ where $h$ is a positive multivariate polynomial, $p$ is a multivariate polynomial with order $m\geq 2$ and the highest order terms $p_m(\theta) \rightarrow \infty $ as $\lVert \theta\rVert \rightarrow \infty$. See \cite{roberts1996geometric} and \cite{jarner2000geometric} for more discussions. 

When the state space is $\bR$, ($\mathcal A1$) can be replaced by a much weaker and essentially necessary assumption, though ($\mathcal A2$) needs to be strengthened a little bit.  We will say that a Metropolis--Hastings chain on $\bR$  satisfies assumption ($\tilde{\mathcal A}$) if it:

\begin{itemize}
	\itemtaone  \label{ass:1dexpo tail}   has a target density $\pi$ which is positive and  there exists some  constant $\alpha > 0$ and $x_1 > 0$ such that, for all $y > x > x_1$:
	\[
	\log \pi(x) - \log \pi(y)  \geq \alpha (y - x),
	\]
	and for all $ y < x < -x_1$:
	\[
	\log \pi(x) - \log \pi(y)  \geq \alpha (x - y).
	\]
	
	\itemtatwo  \label{ass: 1drandom walk}has a random-walk proposal density $q$, that is, $q(\theta, \theta') = q(\theta',\theta) = q(\lvert\theta-\theta'\rvert)$. Furthermore, with $\alpha$ defined as above, there exists a finite $b$ such that $q(s)\leq b e^{-\alpha s}$ for every non-negative $s$.
	\end{itemize}

Assumption $(\tilde A)$ is the condition in Theorem 3.2 of Mengersen and Tweedie \cite{mengersen1996rates}. It is shown that every random-walk Metropolis--Hastings chain on $\bR$ is geometrically ergodic providing it satisfies $(\tilde{\mathcal A})$. Assumption $(\tilde A)$  covers many posterior distributions as Gaussian, Gamma,
exponential from typically-used statistical models.

Now we are ready to provide   a sufficient condition for the geometric ergodicity of $\pex$ 
\begin{theorem}\label{thm:Geometric ergodic inhirit}
	Let $\pmh$ be a random-walk Metropolis--Hastings chain on $\Theta$ with posterior distribution $\pi(\theta|x)$ as stationary distribution. Suppose  $\pmh$ satisfies Assumption $(\mathcal A)$ if $\Theta = \bR^d$ with $d>1$, or Assumption $(\tilde{\mathcal A})$ if $\Theta = \bR$. Furthermore,  assume there exists a  continuous function  $c(s)$ from $\bR_{\geq 0}$ to $[0,1]$ with $c(0) = 0$ such that
	\[
	\lVert p_\theta - p_{\theta'} \rVert_\TV \leq c(\lVert \theta - \theta'\rVert),
	\]
	where $p_\theta$ is the model's likelihood.	Then the exchange chain $\pex$ is also geometrically ergodic.
\end{theorem}

\begin{proof}
	We will assume  $\Theta = \bR^d$ with $d > 1$ as the case $d = 1$ can be proved in  the same way. 	First, for any $\theta\neq \theta'$, straightforward calculation gives
	\begin{align*}
	\dpex(\theta,\theta')  &= q(\theta,\theta') \bE_{p_{\theta'}}{\min\{\frac{\pi(\theta')q(\theta',\theta)f_{\theta'}(x)}{\pi(\theta)q(\theta,\theta')f_{\theta}(x)}\cdot \frac{f_\theta(w)}{f_{\theta'}(w)},1\}} \\
	& = q(\theta,\theta')\bE_{p_{\theta'}}{\min\{\frac{\pi(\theta'|x)q(\theta',\theta)}{\pi(\theta|x)q(\theta,\theta')}\cdot \frac{p_\theta(w)}{p_{\theta'}(w)},1\}} \\
	&\geq q(\theta,\theta')\bE_{p_{\theta'}}{\min\{\frac{\pi(\theta'|x)q(\theta',\theta)}{\pi(\theta|x)q(\theta,\theta')},1\}} \min
	\{ \frac{p_\theta(w)}{p_{\theta'}(w)},1\}\\
	&= \dpmh(\theta,\theta') (1 - \lVert p_\theta - p_{\theta'} \rVert_\TV)\\
	&\geq \dpmh(\theta,\theta')(1 - c(\lVert\theta - \theta'\rVert)).
	\end{align*}
	Next, let $\tilde{P}_{\text{MH}}$ be another random walk MH chain with proposal density $\tilde q$ proportional to $(1-c)q$, that is, 
	$\tilde	q(\theta,\theta') = \frac{ (1-c(\lVert \theta -\theta'\rVert)) q(\lVert \theta - \theta' \rVert)}{C_q}$ where $C_q$ is the normalizing constant that does not depend on $\theta, \theta'$. 
	
	Then we will   check  that $\tilde{P}_{\text{MH}}$  also satisfies $(\mathcal A)$. It is clear that $\tilde{P}_{\mathrm{MH}}$ still satisfies $(\mathcal A1)$ as the target distribution is still $\pi$. To check  $(\mathcal A2)$, the above expression  shows $\tilde q$ is still symmetric. Meanwhile, since there exists some $\delta_q > 0$ and $\epsilon_q >0$ such that $q(\theta_1, \theta_2)\geq \epsilon_q$ for $\Vert \theta_1 - \theta_2 \rVert \leq \delta_q$,
	we can take $\tilde \delta_q$  so small such that $\tilde \delta_q \leq \delta_q$ and $ c(\lVert s\rVert) \leq 1/2$ for every $\lVert s \lVert \leq \tilde\delta_q$. Therefore we have
	$
	\tilde q(\theta, \theta') \geq \tilde \epsilon_q
	$
	when $\lVert \theta - \theta' \lVert \leq \tilde\delta_q$,
	where $
	\tilde \epsilon_q := \frac{\epsilon_q}{2C_q}.$
	
	Therefore the new MH chain with transition kernel $\tilde{P}_{\mathrm{MH}}$
	satisfies $(\mathcal A)$ and it follows directly from Theorem 4.3 in Jarner and Hanson \cite{jarner2000geometric} that   $\tilde{P}_{\text{MH}}$ is also geometrically ergodic.  
	
	Now we compare $\pex$ with $\tilde{P}_{\mathrm{MH}}$. Our previous calculation shows, for any $\theta\neq \theta'$:
	\[
	\dpex(\theta,\theta') \geq \dpmh(\theta,\theta')(1 - c( \lVert\theta - \theta'\rVert)) = C_q \tilde{p}_{\text{MH}}(\theta,\theta').
	\]
	
	When $\theta = \theta'$,
	\[
	\tilde{P}_{\mathrm{MH}}(\theta,\{\theta\}) = 1 - \frac{\int (1-c(\lVert\theta-\theta'\rVert)) q(\lVert\theta-\theta'\rVert)\min\{\pi(\theta'|x)/\pi(\theta|x),1\}d\theta'}{C_q},
	\]
	we have:
	\begin{align*}
	C_q \tilde{P}_{\mathrm{MH}}(\theta,\{\theta\}) &=
	\int_{\bR^d} (1-c(\lVert\theta-\theta'\rVert)) q(\lVert\theta-\theta'\rVert) (1-\min\{\pi(\theta'|x)/\pi(\theta|x),1\})d\theta'\\
	&\leq \int q(\lVert\theta-\theta'\rVert) (1-\min\{\pi(\theta'|x)/\pi(\theta|x),1\})d\theta'\\
	&= \pmh(\theta,\{\theta\}) \leq \pex(\theta,\{\theta\}).
	\end{align*}
	Therefore, for any $\theta,\theta'$, we have
	$\pex(\theta,A) \geq C_q \tilde{P}_{\mathrm{MH}}(\theta,A)$ for every $\theta$ and measurable set $A$. We conclude $\pex$ is also geometrically ergodic by Lemma \ref{lem:uniform likelihood ratio on spectrum}.
\end{proof}

\begin{remark}
	Let $(\mathcal P(\bR^d), \TV)$ be the set of all the probability measures on $\mathbb{R}^d$, equipped with total variation metric. Then the key condition $\lVert p_\theta - p_{\theta'} \rVert_\TV \leq c(\lVert \theta - \theta'\rVert)$ is essentially requiring the map $T: \mathbb R^d \rightarrow \mathcal P(\bR^d)$ being uniformly continuous.
\end{remark}

\subsubsection{A Central Limit Theorem (CLT) for the exchange algorithm}
Let $X = \{X_1, \cdots, X_n, \cdots \}$ be a Markov chain starting from stationary distribution $\pi$. Let $h\in L^2_0(\pi)$, that is, $\bE_\pi(h) = 0$ and $\bE_\pi(h^2)< \infty$. We say a $\sqrt n$-CLT exists for $(h, X)$ if:
\[
\frac{\sum_{i= 0}^n h(X_i)}{\sqrt n} \rightarrow \Nor(0, \sigma^2(X,h)),
\]
for some $ \sigma^2(X,h) < \infty$. Furthermore, we call $ \sigma^2(X,h)$ the asymptotic variance. The CLT and asymptotic variance for general Markov chains is studied extensively in \cite{kipnis1986central}, \cite{mira2001ordering} and \cite{roberts2008variance}. The next theorem gives a CLT for $\pex$, as well as bounds for its asymptotic variance. Before stating the theorem, we briefly review some  preliminary results in spectral theory, further discussions can be found in \cite{chan1994discussion}, \cite{kipnis1986central} and \cite{conway2019course}. Let $H$ be a Hilbert space. Let $M$ be a self-adjoint operator with spectrum $\sigma(M)$ from $H$ to $H$. The spectral theorem guarantees that $M$ has an associated spectral measure $\calE^M$, i.e., a map from Borel subsets of $\sigma(M)$ to self-adjoint projection operators on $M$. Moreover, let $v$ be an element in $H$, the spectral measure  $\calE^M$ further induces a  Borel measure $\calE^M_v$ on $\sigma(M)$ defined as:
\[
\calE^M_v(B) := (\calE^M(B)v, v).
\]
In the following theorem, the Hilbert space is taken as $L^2_0(\pi(\cdot|x))$ and the self-adjoint operator is taken as $\pex$.
\begin{theorem}\label{thm:CLT for exchange}
	Let $\pex$ be a geometrically ergodic Markov chain with stationary distribution $\pi(\theta|x)$. Then we have
	\begin{itemize}
		\item  For any $h \in L^2_0(\pi(\cdot|x))$, a $\sqrt n$-CLT exists for $(h, \pex)$.  
		\item The asymptotic variance $ \sigma^2(\pex,h)$ has the following  representation:
		\begin{equation}\label{eqn:asymptotic variance}
		\sigma^2(\pex,h) = \int_{-1}^1 \frac{1+ \lambda}{1 - \lambda} \mathcal E^{\text{EX}}_h (d\lambda),
		\end{equation}
		where $ \mathcal E^{\text{EX}}_h (d\lambda)$ is the Borel measure induced by $h$ and the spectral measure $ \mathcal E^{\text{EX}}$.
		\item The relationship between $\sigma^2(\pex,h)$ and  $\sigma^2(\pmh,h)$ is given by:
		\begin{equation}\label{eqn:inequality av}
		\sigma^2(\pmh,h) \leq \sigma^2(\pex,h)	\leq \frac {1 - m(\pmh)}{1 + m(\pmh)}\frac{2}{1 - M(\pex)}\sigma^2(\pmh,h)
		\end{equation}
		for any $h \in L^2_0(\pi(\cdot|x))$, where $m(\pmh)$ is the infimum of the spectrum of $\pmh$  on $L^2_0(\pi(\cdot|x))$. Similarly,  $M(\pex)$ is the supremum of the spectrum of  $\pex$  on $L^2_0(\pi(\cdot|x))$. Notice that this bound does not depend on $h$.
	\end{itemize}
	
\end{theorem}
\begin{proof}[Proof of Theorem \ref{thm:CLT for exchange}] 
	Let $\mathcal E^\mathrm{EX}$ be the spectral measure for $\pex$ and let $\calE^\mathrm{EX}_h$ be the induced Borel measure.
	The existence of $\sqrt n$-CLT and the expression for the asymptotic variance follows from \cite{kipnis1986central}, page 3 and \cite{geyer1992practical}, Theorem 2.1. Furthermore,  $\pex$ is dominated by $\pmh$ in the Peskun's order as shown in Theorem \ref{thm:lazy MH chain}, it implies the operator $\pex - \pmh$ is positive. Moreover, it follows from Theorem 4 in \cite{tierney1998note} that
	$$\sigma^2(\pmh,h) \leq \sigma^2(\pex,h).$$ 
	
	The right part of inequality \ref{eqn:inequality av} is established by writing  both $\sigma^2(\pex,h)$ and $\sigma^2(\pmh,h)$ as integration with respect to their spectral measures and using the facts 
	$$\sigma^2(\pex,h) \leq \frac{2}{1 - M(\pex)}\bE_{\pi(\cdot|x)}(h^2)$$ and 
	$$\sigma^2(\pmh,h) \geq \frac{1+ m(\pmh)}{1 - m(\pmh)}\bE_{\pi(\cdot|x)}(h^2).$$
	
\end{proof}

Theorem \ref{thm:lazy MH chain}, \ref{thm:Geometric ergodic inhirit} and \ref{thm:CLT for exchange} gives theoretical results of exchange algorithm. Now we will apply these theorems (especially Theorem \ref{thm:Geometric ergodic inhirit}) to practically useful models. 

\subsection{Practical applications: location models, Poisson models, and exponential families}\label{subsec:applications of general state convergence}
As we will see in this section, the condition $
\lVert p_\theta - p_{\theta'} \rVert_\TV \leq c(\lVert \theta - \theta'\rVert)$ is satisfied by a large number of models with unbounded parameter spaces.

\begin{example}[Location models]
	Consider a location family with $p_\theta(x) = p(x - \theta)$ be the location families  with $\theta\in \bR$. Then it is clear that:
	\[
	\lVert p_\theta - p_{\theta+s} \rVert_\TV  = \lVert p_0 - p_{s} \rVert_\TV.
	\]
	Then we can define $c(s) = \lVert p_0 - p_{s} \rVert_\TV$ which satisfies our condition. 
	
	In particular, let $p_\theta \sim \Nor(\theta, 1)$ be a family of Gaussian distributions with unknown mean $\theta$ and known variance $1$. If we  put a Gaussian prior $\pi(\theta) \sim \Nor(0, \sigma^2)$ on $\theta$, it is clear that the posterior distribution is also Gaussian and thus has exponential tails. Therefore, by Theorem \ref{thm:Geometric ergodic inhirit}, the random-walk exchange algorithm for the posterior is geometrically ergodic.
\end{example}

\begin{example}[Poisson model]
	Consider a Poisson family with mean parameter $\theta$:
	\[\Prob_\theta \sim \Poi(\theta).\]
	We claim that 
	\[\lVert p_\theta - p_{\theta+s}  \rVert_\TV \leq 1 - e^{-s} \leq s,\]	
	therefore we can choose $c(s) = 1 - e^{-s}$ or $c(s) = \min\{1,x\}.$
	
	The proof of our claim uses a simple coupling argument. Let $X,Y$ be independent Poisson random variables with parameter $\theta, s$ respectively. Let $Z = X + Y$ and it is clear that $Z\sim \Poi(s + \theta)$. On the other hand, 
	\[
	\lVert p_\theta - p_{\theta+s}  \rVert_\TV \leq \Prob(X\neq Z) = \Prob(Y\neq 0) = 1 - e^{-s},
	\]
	which proves the first part of our inequality, the second part is the standard inequality.

	Now it suffices to check assumption $(\tilde{\mathcal A}1)$, i.e., the posterior distribution has an exponential tail. We can show that, if the prior density decays for large $\lambda$, the posterior distribution will satisfy  $(\tilde{\mathcal A}1)$.
	
	\begin{proposition}\label{prop: poisson likelihood}
		Let $\pi(\lambda) $ be a prior density on $[0,\infty)$, assume there exists $\lambda_0 > 0$  such that
		$$\pi(\lambda') < \pi(\lambda)   \text{ for any } \lambda' >\lambda > \lambda_0,$$
		then the posterior density with Possion likelihood satisfies  $(\tilde{\mathcal A}1)$.
	\end{proposition}
	
	\begin{proof}
		The posterior density has the following form:
		$
		\pi(\lambda|x) \propto \pi(\lambda) e^{-\lambda} \lambda^x,
		$
		for any $\lambda' > \lambda > 0$:
		\[
		\log \pi(\lambda|x) - \log\pi(\lambda'|x) = (\lambda' - \lambda) + x (\log \lambda - \log \lambda') + (\log(\pi(\lambda)) - \log(\pi(\lambda')))
		\]
		
		If we take $\lambda' > \lambda > \max\{\lambda_0, 2x\}$, then
		\[
		\log \pi(\lambda|x) - \log\pi(\lambda'|x)  \geq (\lambda' - \lambda) - \frac{x}{\lambda} (\lambda' - \lambda) \geq \frac 12 (\lambda' - \lambda),
		\]
		as desired.
		
	\end{proof}
	
	Proposition \ref{prop: poisson likelihood} ensures for any prior density with finitely many modes, the random walk Metropolis--Hastings algorithm and the corresponding random-walk exchange algorithm 
	are both geometrically ergodic. This includes  many practical prior distributions such as:
	\begin{itemize}
		\item Conjugate prior (Gamma distribution)
		\item Any finite mixture of Gamma distributions
		\item Any normal prior $ \Nor(\mu, \sigma^2)$ truncated at $[0,\infty)$.
	\end{itemize}
\end{example}

Besides the above two examples, Theorem \ref{thm:Geometric ergodic inhirit} can be applied to a large subset of exponential families. 
Consider the exponential family with density (or probability mass) of the  form
$p_\theta(x) = h(x)e^{\theta \cdot T(x) - \eta(\theta)}.$ Here $\theta$ is often referred to as `canonical parameter', and statistics $T(X)$ is often referred to as `sufficient statistics'. To fix ideas, we allow $x$ be discrete or continuous, one-dimensional or multi-dimensional, but we assume the canonical parameter to be a one-dimensional parameter. The set of parameters $\theta$ for which the integral (or summation) below  is finite is referred to as the natural parameter space:
\[
\mathcal N := \{\theta: \int h(x) e^{\theta\cdot T(x)} dx < \infty \}.
\]

Let $d_\KL(\theta,\theta')$ be the  Kullback--Leibler (K--L) divergence between $p_\theta$ and $p_{\theta'}$. For discrete cases, K--L divergence is defined by:
\[
d_\KL(\theta,\theta') = \sum_{x\in \mathcal X} - p_\theta(x) \log\bigg(\frac{p_{\theta'}(x)}{p_\theta(x)} \bigg).
\]
For continuous cases, K--L divergence is defined by:
\[ 
d_\KL(\theta,\theta') = \int_{x\in \mathcal X} - p_\theta(x) \log\bigg(\frac{p_{\theta'}(x)}{p_\theta(x)} \bigg)dx.
\]
The next theorem shows, when the sufficient statistics is uniformly bounded, then the condition $\lVert p_\theta - p_{\theta'} \rVert_\TV \leq c(\lvert \theta - \theta'\rvert)$ is satisfied.
\begin{theorem}\label{thm:exchange inherit,finite exponential family}
	Let $\{p_\theta\}$ be an exponential family. If there exists $M>0$ such that the sufficient statistics $T(x)$ satisfies $\lvert T(x) \rvert\leq M$ almost everywhere under any $p_\theta$. Then we have:
	\[
\lVert p_\theta - p_{\theta'} \rVert_\TV \leq \frac12 \sqrt{d_\KL(\theta,\theta') + d_\KL(\theta',\theta)} \leq \frac{\sqrt 2M}{2} \sqrt{\lvert \theta - \theta'\rvert}.
	\]
\end{theorem}
\begin{proof}
	The first inequality is generally true for any two distributions. Given two probability distribution $P,Q$, Pinsker's inequality \cite{tsybakov2008introduction} says:
	$
	\lVert P - Q \rVert_\TV^2 \leq \frac 12 d_\KL(P,Q),
	$
	swaping the order of $P,Q$ and use Pinsker's inequality again gives the first inequality. 
	
	For an exponential family with discrete sample space, the KL divergence $d_\TV(\theta,\theta')$ can be written as:
	\begin{align*}
	d_\KL(\theta,\theta') &= \sum_{x\in \mathcal X} - p_\theta(x) (\theta' \cdot T(x) + \eta(\theta) - \theta \cdot T(x) - \eta(\theta')) \\
	& = -(\theta' - \theta) \bE_\theta(T) - \eta(\theta) + \eta(\theta'),
	\end{align*}
	similarly, $d_\KL(\theta',\theta)$ equals:
	\begin{align*}
	d_\KL(\theta',\theta) = -(\theta - \theta') \bE_{\theta'}(T) - \eta(\theta') + \eta(\theta).
	\end{align*}
	Hence we have:
	\[
	d_\KL(\theta,\theta') + d_\KL(\theta',\theta) =  (\theta' - \theta) \big(\bE_{\theta'}(T) - \bE_{\theta}(T)\big) \leq 2M \lvert \theta - \theta'\rvert,
	\]
	as $\lvert T(x)\rvert$ is uniformly bounded by $M$, which proves the second inequality. 
	
	For a continuous sample space, we just change all the summation above by integration and all the results still hold.
\end{proof}
\begin{remark}
	A standard result in exponential families is $
	\bE_{\theta}(T)  = \eta'(\theta)$.
	Therefore the condition $\lvert T\rvert \leq M$ can be relaxed by $\eta'$  is Lipschitz continuous, or the second-order derivative of $\eta$ is bounded.
\end{remark}

Exponential family includes a large number of practical models. In particular, two examples mentioned at the beginning of this paper: Ising model (Example \ref{eg:ising}) and exponential family graph model (ERGM) (Example \ref{eg:ergm}), these two examples both belong to the exponential family with bounded sufficient statistics.

\begin{example}[Ising Model, revisited]\label{eg:ising revisited}
	With all the definitions the same as in Example \ref{eg:ising}. An Ising model  is defined as the following probability distribution over all possible configurations on a graph $G = (V,E)$:
	\begin{equation}\label{eqn:ising}
	\Prob_\theta(\sigma) = \frac{e^{-\theta H(\sigma)}}{Z(\theta)}.
	\end{equation}
The sufficient statistics is uniformly bounded as there are only finitely many possible spin configurations.
\end{example}

\begin{example}[Exponential Random Graph Model, revisited]\label{eg:ergm revisited}
With all the definitions the same as in Example \ref{eg:ergm}.  An Exponential Random Graph Model is defined as the following probability distribution on $\mathcal G_n$, the set of all graphs with $n$ vertices:
	\begin{equation}\label{eqn:ergm}
	\Prob_\theta(G = g) = \frac{e^{\theta s(g)}}{Z(\theta)},
	\end{equation}     
where $s$ is the sufficient statistics. Again, the sufficient statistics is uniformly bounded as there are only finitely many possible spin configurations.
\end{example}

Combining Theorem \ref{thm:Geometric ergodic inhirit} and Theorem \ref{thm:exchange inherit,finite exponential family}, we have the following:
\begin{theorem}\label{thm:practical exchange}
	Let $\pi(\theta|x)$ be the posterior distribution given by prior $\pi(\theta)$ and likelihood satisfying the assumption of Theorem \ref{thm:exchange inherit,finite exponential family}. Furthermore let the original Metropolis--Hastings chain satisfies assumption $(\tilde{\mathcal A})$. The induced exchange chain $\pex$ is geometrically ergodic. 
\end{theorem}
As both of the Ising model and ERGM are defined on discrete (though very large) sample space, the sufficient statistics is uniformly bounded by nature. The next corollary is immediate:
\begin{corollary}
	Suppose the exchange algorithm has an ERGM or Ising likelihood, a random-walk proposal kernel, a uniformly exponential or lighter posterior density, then the corresponding Markov chain is geometrically ergodic. 
\end{corollary}
Theorem \ref{thm:practical exchange} shows one only needs to check Assumption $(\tilde{\mathcal A})$ for the original Metropolis--Hastings algorithm.  The random-walk proposal is by design of the algorithm, so it suffices to verify $(\tilde{\mathcal A}1)$ for the posterior distribution. The next corollary shows, under a Gaussian prior (which is the most popular choice for a prior distribution on unbounded parameter space) the posterior distribution has uniformly exponential or lighter tail.

\begin{corollary}\label{cor:gaussian prior}
	The posterior distribution with a Gaussian prior $ \Nor(\mu, \sigma^2)$ and exponential family likelihood $p_\theta$ has tail lighter than exponential. In particular, the posterior distribution satisfies the assumption $(\tilde{\mathcal A}1)$.
\end{corollary}

The proof of Corollary \ref{cor:gaussian prior}  is included in the proof of the next proposition as a special case. The next proposition shows, if the prior density satisfies $(\tilde{\mathcal A}1)$, then so does the posterior density.

\begin{proposition}\label{prop: exponential family condition}
	Let $\pi(\theta)$ be a prior distribution on $\bR$ satisfying $(\tilde{\mathcal A}1)$, let $p_\theta(x) = h(x)e^{\theta\cdot T(x) - \eta(\theta)}$ be the probability mass/density
function for an exponential family on $\calX$ with base measure $\mu$. Assume the natural parameter space $
 \{\theta: \int h(x) e^{\theta\cdot T(x)} \mu(dx) < \infty \} = \bR$ and the  sufficient statistics is uniformly bounded,	then the posterior distribution also satisfies $(\tilde{\mathcal A}1)$.
\end{proposition}
\begin{proof}
	
	We define a new measure $\mu_1(dx) := h(x)\mu(dx)$, which is a finite measure  on $\mathcal X$ with $\mu_1(\calX) = e^{\eta(0)}$. Let $M_1$ be $\mathrm{esssup}_x T(x)$ and $m_1$ be $\mathrm{essinf}_x T(x)$ with respect to  $\mu_1$. We claim  $\eta'(\theta)   =\bE_\theta(T(x))$ goes to $M_1$  as $\theta$ goes to  $\infty$, and $m_1$ as $\theta$ goes to  $-\infty$.

	Let $S_{c} = \{x: T(x) \leq c\}$, and we will show $p_\theta(S_{M_1-\epsilon}) \rightarrow 0$ as $\theta \rightarrow \infty$ for every $\epsilon > 0$, which in turn shows $\bE_\theta(T(x))\rightarrow M_1$ as $\theta \rightarrow \infty$, here we slightly abuse the notation and use $p_\theta$ to denote both the probablity density function and the probability measure. 
	Notice that for every $\theta > 0$, we can upper bound 
	 $p_\theta(S_{M_1-\epsilon})$ by
	\begin{align*}
		p_\theta(S_{M_1-\epsilon}) = \frac{\int_{S_{M_1-\epsilon}} e^{\theta T(x)}\mu_1(dx)}{\int_{\cal X} e^{\theta T(x)}\mu_1(dx)} 
		&\leq \frac{\int_{S_{M_1-\epsilon}} e^{\theta T(x)}\mu_1(dx)}{\int_{{\cal X} \setminus S_{M_1-\epsilon/2}} e^{\theta T(x)}\mu_1(dx)}\\
		& \leq \frac{\mu_1(S_{M_1 - \epsilon})}{1 -  \mu_1(S_{M_1 - \epsilon/2})} \frac{e^{\theta(M_1 - \epsilon)}}{e^{\theta(M_1 - \epsilon/2)}},
	\end{align*}
	where the last inequality follows from the fact that $e^{
	\theta T(x)} \leq e^{\theta(M_1 - \epsilon)}$ on $S_{M_1-\epsilon}$ and  $e^{
	\theta T(x)} \geq e^{\theta(M_1 - \epsilon/2)}$ on $\calX \setminus S_{M_1 - \epsilon/2}$. The definition of essential supremum guarantees the denominator $1 -   \mu_1(S_{M_1 - \epsilon/2})$ is strictly positive. Therefore $p_\theta(S_{M_1 - \epsilon}) \rightarrow 0$ as $\theta \rightarrow \infty$, as desired. Similarly $\bE_\theta(T(x))$ goes to $m_1$  when $\theta$ goes to  $-\infty$.
		
	Given a prior density $\pi(\theta)$, the posterior distriubution can be formally written as:	
	\[
	\pi(\theta|x) \propto \pi(\theta) e^{\theta\cdot T(x) - \eta(\theta)}.
	\]
	For $\theta' > \theta > 0$, we have
	\begin{align*}
	\log \pi(\theta|x) - \log\pi(\theta'|x)   &= \log(\pi(\theta)) - \log(\pi(\theta')) + (\theta - \theta') T(x)  - (\eta(\theta) - \eta(\theta'))\\
	& = 	\log(\pi(\theta)) - \log(\pi(\theta'))  + (\theta  - \theta') (T(x) - \eta'(\xi))
	\end{align*}
	where $\xi \in [\theta, \theta']$.
	
	By assumption,  there exists $\theta_0 >0$ and $\alpha >0$ such that the first term $	\log(\pi(\theta)) - \log(\pi(\theta')) \geq \alpha (\theta' - \theta)$ for any $\theta' > \theta >\theta_0$. We can further choose $\epsilon = \frac \alpha 2$ and $\theta_1 > 0$ such that  $\eta'(\theta) > M_1  - \epsilon$ when $\theta > \theta_1$. Then, for any $\theta' > \theta >\max\{\theta_0,\theta_1\}$:
	\[
	\log \pi(\theta|x) - \log\pi(\theta'|x) \geq \frac \alpha  2 (\theta' - \theta),
	\]
	as desired. The proof for $\theta \in (-\infty, 0)$ is essentially the same. 
	
\end{proof}

Proposition \ref{prop: exponential family condition} includes many practical prior distributions, for example:
\begin{itemize}
	\item Any Gaussian prior (as discussed in Corollary \ref{cor:gaussian prior})
	\item Any finite mixture of Gaussian priors
	\item Conjugate prior $\pi_{n_0, t}(\theta) \propto e^{n_0(\theta t - \eta(\theta))},$
	with $n_0 > 0$ and $t \in (m_1, M_1)$
	\item Any finite mixture of conjugate priors:	
	$\pi(\theta)\propto \sum_{i=1}^n \lambda_i  \pi_{n_{0,i}, t_i},$
	with $n_{0,i} >0$, $t_i \in (m_1,M_1)$ for any $i$.
\end{itemize}

If the parameter space is artificially defined as $[0,\infty)$ or $(0, \infty)$ (for example, in Ising model, $\theta$ corresponds to the `inverse temperature' and is thus always positive), then Proposition \ref{prop: exponential family condition} still includes the `truncated version' of all  the  models mentioned above:

\begin{itemize}
	\item Any Gaussian prior (as discussed in Corollary \ref{cor:gaussian prior}) truncated at $[0,\infty)$
	\item Any finite mixture of truncated Gaussian priors
	\item Conjugate prior $\pi_{n_0, t}(\theta) \propto e^{n_0(\theta t - \eta(\theta))},$
	with $n_0 > 0$ and $t <  M_1$
	\item Any finite mixture of conjugate priors: 
	$\pi(\theta)\propto \sum_{i=1}^n \lambda_i  \pi_{n_{0,i}, t_i},$
	with $n_{0,i} >0$, $t < M_1$ for any $i$.
\end{itemize}

To summarize, this part concentrates on bridging the gap between theoretical results and  practical applications of the exchange algorithm. Our results guarantee that, under mild conditions, the exchange algorithm used in real applications is geometrically ergodic.  We hope this positive result will give practitioners `peace of mind' when applying the exchange algorithm.

\section{Conclusion}\label{sec:conclusion}
To summarize, the first part of our results focuses on analyzing the asymptotic variance of the exchange chain. The second part focuses on convergence speed. Our results justify the the theoretical usefulness of the exchange algorithm. When the original Metropolis--Hastings algorithm satisfies assumption $(\mathcal A)$ or $(\tilde{\mathcal A})$ and the likelihood function satisfies the assumption of Theorem \ref{thm:Geometric ergodic inhirit}, the exchange algorithm is proven to be geometrically ergodic and admits a $\sqrt n$-CLT for any square-integrable function. In particular, assumption $\mathcal A$  or $(\tilde{\mathcal A})$  is naturally satisfied in many practical applications, including but not limited to location models,  Ising models, and ERGMs. It is our hope that this paper can be used to fill some gaps between the Markov chain Monte Carlo theory and applications.

However, a lot more has to be done. The convergence analysis here is mostly based on spectral theory and is unable to provide  so-called `honest' bounds. That is, our results show the exchange algorithm converges to the stationary distribution at a geometric rate but does not give practical bounds on the rate.  For general Metropolis--Hastings algorithms, convergence rates are usually established by using the  `drift-and-minorization' approach of Rosenthal \cite{rosenthal1995minorization}. It is an outstanding open problem to establish a drift and minorization condition for the underlying exchange chain. Moreover, the exchange algorithm can be included in the  framework developed by Andrieu et al.\cite{andrieu2018utility}, therefore it would be interesting to investigate if the main results of this paper can be generalized to the general framework.

 Even though one could establish the drift and minorization conditions and get `honest' bounds, usually they are still far away from `practical bounds'. It would be a more ambitious project to sharpen the rates of convergence derived from the `drift-and-minorization' framework, which would be of independent interest and is beyond the scope of exchange algorithm. Admittedly, there is always a gap between theory and practice, but we hope that more `practical' theories can be established to fill this gap.

\section*{Acknowldegement}

The author would like to thank Persi Diaconis, Julia Palacios, Wing H.Wong, and Daniel Rudolf for helpful discussions and comments. The author would like to thank the Editor, the Associate Editor and two referees for their constructive suggestions.

%%%%%%%%%%%%%%%%%%%%%%%%%%%%%%%%%%%%%%%%%%%%%%
%% Supplementary Material, if any, should   %%
%% be provided in {supplement} environment  %%
%% with title and short description.        %%
%%%%%%%%%%%%%%%%%%%%%%%%%%%%%%%%%%%%%%%%%%%%%%

%%  MR numbers will be added by VTeX.                      %%
%%                                                         %%
%%  Use \cite{...} to cite references in text.             %%
%%                                                         %%
%%%%%%%%%%%%%%%%%%%%%%%%%%%%%%%%%%%%%%%%%%%%%%%%%%%%%%%%%%%%%

%% if your bibliography is in bibtex format, uncomment commands:
\bibliographystyle{imsart-number} % Style BST file (imsart-number.bst or imsart-nameyear.bst)
\bibliography{bibliography}       % Bibliography file (usually '*.bib')

\end{document}